\documentclass[journal]{IEEEtran}
\usepackage{amsfonts}
\usepackage{cite,graphicx,amsmath,amsthm}
\usepackage{subfigure}
\usepackage{fancyhdr}
\usepackage{dsfont}
\usepackage{array,color}
\usepackage{bm}
\usepackage{float}
\usepackage{algorithm}
\usepackage{algpseudocode}
\usepackage{multirow}
\usepackage{booktabs}
\usepackage{multirow}
\usepackage{makecell}

\allowdisplaybreaks[4]

\newtheorem{lemma}{Lemma}

\newtheorem{proposition}{Proposition}

\newtheorem{corollary}{Corollary}

\newtheorem{property}{Property}

\newtheorem{remark}{Remark}

\newtheorem{claim}{Claim}

\newtheorem{theorem}{Theorem}

\def\BibTeX{{\rm B\kern-.05em{\sc i\kern-.025em b}\kern-.08em
    T\kern-.1667em\lower.7ex\hbox{E}\kern-.125emX}}

\newcommand{\tabincell}[2]{\begin{tabular}{@{}#1@{}}#2\end{tabular}}

\begin{document}

\title{{Learning Centric Wireless Resource Allocation for Edge Computing: Algorithm and Experiment}}

\author{Liangkai Zhou, Yuncong Hong, Shuai Wang, Ruihua Han, Dachuan Li, Rui Wang, and Qi Hao   
\thanks{
Liangkai~Zhou, Yuncong~Hong, and Rui~Wang are with the Department of Electrical and Electronic Engineering, Southern University of Science and Technology (SUSTech), Shenzhen 518055, China (e-mail: 11510196@mail.sustech.edu.cn, hongyc@mail.sustech.edu.cn, wang.r@sustech.edu.cn).

Shuai~Wang is with the Department of Electrical and Electronic Engineering, and is also with the Department of Computer Science and Engineering, Southern University of Science and Technology (SUSTech), Shenzhen 518055, China. (e-mail: wangs3@sustech.edu.cn)

Ruihua~Han and Dachuan~Li are with the Department of Computer Science and Engineering, Southern University of Science and Technology (SUSTech), Shenzhen 518055, China (e-mail: hanruihuaff@gmail.com, dachuanli86@gmail.com).

Qi~Hao is with the Department of Computer Science and Engineering, Southern University of Science and Technology (SUSTech), Shenzhen 518055, China, and is also with the Sifakis Research Institute of Trustworthy Autonomous Systems, Shenzhen, China (e-mail: hao.q@sustech.edu.cn).

}
}
\maketitle

\begin{abstract}
Edge intelligence is an emerging network architecture that integrates sensing, communication, computing components, and supports various machine learning applications,
where a fundamental communication question is: how to allocate the limited wireless resources (such as time, energy) {to the simultaneous model training of heterogeneous learning tasks}?
Existing methods ignore two important facts: 1) different models have {heterogeneous demands on training data}; 2) there is a mismatch between the simulated environment and the {real-world} environment.
As a result, they could lead to low learning performance in practice.
This paper proposes the learning centric wireless resource allocation (LCWRA) scheme that maximizes the worst learning performance of multiple tasks.
{A}nalysis shows that the optimal transmission time has an inverse power relationship with respect to the generalization error.
Finally, both simulation and experimental results are provided to verify the performance of the proposed LCWRA scheme {and its robustness in real implementation}.
\end{abstract}


\IEEEpeerreviewmaketitle
\section{Introduction}

\IEEEPARstart{E}dge {computing has been proposed in recent years as a solution to the high transmission cost in cloud computing and {limited} performance in local computing \cite{edge1}.} It pushes the cloud services from the network core to the network edges that are in closer proximity to {Internet of Things (IoT)} devices \cite{edge2}.
With the edge computing platform, various machine learning applications can reside at the edge, giving rise to an emerging paradigm termed edge intelligence \cite{edge3}.
Edge intelligence can be divided into edge model training and edge model inference.
In general, there are two ways to implement edge model training: centralized training and distributed training.
Centralized training collects sensing data generated from IoT devices and trains the learning models at the edge \cite{lcpa,lcpa2}.
Distributed training deploys individual learning models at user terminals, and all the users upload their local model parameters periodically to the edge for model aggregation and broadcasting \cite{distributed}.
Since the IoT devices {usually} {can hardly} process the data due to limited computation power, this paper focuses on centralized training.

{In centralized training, a fundamental communication question is: how to allocate the limited wireless resources (such as time, energy) to train multiple learning tasks simultaneously?}
{Existing methods such as the time fairness scheme, the throughput fairness scheme \cite{fairness}, or the importance-aware transmission scheme \cite{importance, importance2} ignore the heterogeneous requests of learning tasks on the data amount, probably leading} to low learning performance.
Furthermore, all these schemes are verified only via computer simulation.
{Due} to the transmission overhead of control signals \cite{overhead} (which results in reduced number of useful data) and random nature of wireless channels \cite{packetloss} (e.g., disconnection/reconnection procedures), there is a mismatch between the simulated environment and the real-world environment.
Therefore, apart from computer simulation, hardware experiments are required to demonstrate the robustness of resource allocation algorithms against practical uncertainties.

To address the above problems, this paper adopts the statistical mechanics of learning \cite{sm} {to predict the learning accuracy} of different tasks {versus the amount of training data,} and reveal the relationship between the learning performance and the wireless resources/channels.
By leveraging the difference of convex programming (DCP), {a} learning centric wireless resource allocation (LCWRA) {maximizing} the learning performance is proposed.
Analysis shows that the optimal transmission time has an inverse power relationship with respect to the generalization error.
This empowers us to use a ranking-based procedure to derive the optimal time when the network energy is unlimited.
Simulation results are provided to verify the performance of the proposed LCWRA scheme.
Experimental results for robotic image dataset collection and vehicular point-cloud dataset collection are further carried out to demonstrate the effectiveness {and robustness} of LCWRA in real systems.

\section{System Model and Problem Formulation}

Consider an edge intelligence system consisting of {one} server and $K$ users with dataset $\left\{\mathcal{D}_{1}, \cdots, \mathcal{D}_{K}\right\}$.
There are $M$ learning models to be trained at the edge, and let $\mathcal{Y}_{m} \in \left \{ \mathcal{Y}_{1} ,\cdots, \mathcal{Y}_{M} \right \}$ denote the users who transmit training samples for the model $m \in \{1, \cdots, M \}$.
The total number of training samples for model $m$ collected at the edge is
\begin{equation}
v_{m} = \sum_{ k \in \mathcal{Y}_{m}} \left\lfloor \frac{d_{k}}{D_k} \right\rfloor +c_{m} \approx \sum_{ k \in \mathcal{Y}_{m}} \frac{d_{k}}{D_k} +c_{m}, \label{v_m}
\end{equation}
where $d_{k}$ represents the data amount transmitted by user $k \in \{ 1,\cdots, K \}$, $D_{k}$ represents the data size per sample, and $c_{m}$ denotes the number of historical samples available at the edge for model $m$. The approximation is based on $\lfloor x \rfloor \approx x$ when $x \gg 1$.
In this paper, it is assumed that the datasets $\left\{\mathcal{D}_{1}, \cdots, \mathcal{D}_{K}\right\}$ are independent and identically distributed.\footnote{{Under non-IID data distribution among users, the upper bound data-rate constraint may be imposed on the users who transmit low-quality data, and the lower bound data-rate constraint may be imposed on the users who transmit high-quality data.}}
Therefore, the generalization error, denoted as $\Psi_{m}$ is a function of the number of samples $v_m$.
To the best of the {authors'} knowledge, there is no exact expression of $\Psi_{m}(v_{m})$.
To this end, an inverse power law model \cite{lcpa, lcpa2, curve1, curve2}, which is supported by statistical mechanics of learning \cite{sm}, is adopted to approximate $\Psi_{m}$ as
\begin{equation}
	\Psi _{m} \approx a_{m}\, v_{m}^{-b_{m}}, \label{Psi_m}
\end{equation}
where $a_{m}, b_{m} > 0 $ are tuning parameters.

Our aim is to minimize $\{\Psi_{m}{|\forall m}\}$ by proper wireless resource allocations.
Specifically, {we} consider the sample transmission in a single-hop time division multiple access (TDMA)-based wireless system.
Within a period of $T_{\rm{max}}$, the user $k$ is assigned a duration of $t_{k} \in \mathbb{R}_{+}$ to transmit samples.
As a result, the data amount transmitted by user $k$ is
\begin{equation}
	d_{k} = t_{k}R_{k}, \label{d_k}
\end{equation}
where $R_{k}$ represents the data-rate of user $k$.
Since the network adopts the TDMA, the transmission of each user is scheduled by a wireless controller and therefore the transmission is collision-free\cite{A1}.\footnote{If the system adopts the carrier sense multiple access with collision avoidance (CSMA-CA) protocol, orderly and collision-free transmissions would be an unrealistic assumption.
In such a case, the proposed learning centric wireless resource allocation framework is still applicable by adding a factor $\alpha$ with $0<\alpha<1$ to $R_k$, where $\alpha$ represents the transmission loss due to collisions.
Moreover, by combining the interference management techniques in \cite{C7} and \cite{C8}, the proposed LCWRA is applicable to other access schemes such as the OFDMA and the NOMA.}
Therefore, $R_k$ is determined by the bandwidth $B$ allocated to the system and the signal-to-noise ratio at user $k$, {i.e.,}
\begin{align}
R_{k} &= B\log_{2} \left( 1+\frac{|h_{k}|^{2}E_{k}}{t_k \sigma^{2}}\right)
, \label{R_k}
\end{align}
where $E_{k}\in \mathbb{R}_{+}$ represents the transmit energy at user $k$,  $h_{k} \in \mathbb{C}$ represents the channel\footnote{If the channel is time-variant, $|h_{k}|^2$ is replaced by its expectation $\mathbb{E}[|h_{k}|^2]$.} between the user $k$ and the edge, and $\sigma^{2}$ denotes the power of Gaussian white noise.
Hence, the transmit power at user $k$ is $E_{k}/t_{k}$.
The transmit power is limited by the hardware structure as $E_{k}/t_{k}\leq P_{\rm{max}}$, where $P_{\rm{max}}$ denotes the peak power.
In applications such as sensor networks, the total energy is usually limited by an upper bound $E_{\mathrm{max}}$ for the consideration of energy efficiency \cite{C9}, which gives $\sum_{k=1}^KE_k\leq E_{\rm{max}}$.

In the considered edge computing system driven by learning tasks, the aim is to minimize generalization errors subject to  communication constraints, which leads to the following optimization problem{.}
\begin{subequations}
	\label{P}
	\begin{align}
	\textrm{P}1: \quad& \min_{\vec{\mathbf{t}}, \vec{\mathbf{e}}}~\max_{m=1,\cdots,M} \nonumber\\
	&\left\{ a_{m}  \left( \sum_{k \in \mathcal{Y}_{m}} \frac{t_{k}B}{D_{k}} \log_{2} \left( 1+ \frac{E_{k} |h_{k}|^{2} }{t_{k} \sigma^{2}} \right) +c_{m} \right)^{-b_{m}} \right\},\nonumber  \\
	\textrm{s.t.} \quad& \sum_{k=1}^{K} t_{k} \leq T_{\textrm{max}}, \quad \sum_{k=1}^{K} E_{k} \leq E_{\textrm{max}}, \label{P_ca} \\
	&\frac{E_{k}}{t_{k}} \leq P_{\rm{max}}, \quad \forall k, \label{P_cc} \\
	&\frac{t_{k}B}{D_{k}}\log_{2} \left( 1+ \frac{E_{k} |h_{k}|^{2} }{t_{k} \sigma^{2} }\right) \leq |\mathcal{D}_{k}|, \quad \forall k, \label{P_cd} \\
	&t_{k} \geq 0, \ E_{k} \geq 0, \quad \forall k, \label{P_ce}
	\end{align}
\end{subequations}
where the min-max operation in the objective function aims to warrant the worst learning performance of all models,  {$\vec{\mathbf{t}}=[t_1,\cdots,t_K]^T$ and $\vec{\mathbf{e}}=[E_1,\cdots,E_K]^T$.}

\textbf{Use Case (Vehicular Point-Cloud Dataset Collection):} For autonomous vehicles, object detection based on artificial intelligence (AI) models is needed for collision avoidance \cite{A5,shaoshuai}.
To train AI models, AI companies (e.g., Google, Baidu) arrange vehicles to collect sensing data in various environment (as shown in Fig.~1).
The time is limited as the sensing data is massive.
Thus data selection and vehicular wireless resource allocation are required.
After data transmission, the data is manually annotated and used for training AI models \cite{A6}.

\begin{figure*}[!t]
	\centering
	\includegraphics[width=0.95\textwidth]{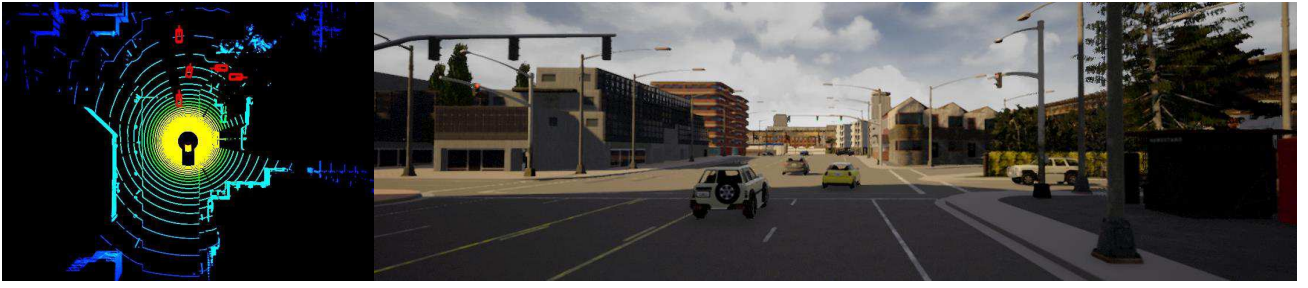}
	\caption{The use case of vehicular point-cloud dataset collection.}
	\label{map}
\end{figure*}

\section{DCP-based LCWRA Algorithm}

{For notation convenience,} we first define function
\begin{align}
&d_{k}=\Theta_k(t_{k},E_{k})=t_kB\log_{2} \left( 1+E_{k} |h_{k}|^{2}/(t_{k} \sigma^{2})\right).
\end{align}
As $\Theta_{k}(t_{k},E_{k})$ is the perspective transformation of logarithm function $B\log_{2} \left( 1+E_{k} |h_{k}|^{2}/\sigma^{2}\right)$ with respect to $t_k$, it is concave {with respect to} both $t_{k}$ and $E_{k}$.
On the other hand, as $a_{m}>0$ and $b_{m}>0$, $a_{m}\,x^{-b_{m}}$ is a decreasing and convex function of $x$.
According to the composition law \cite{opt1}, $a_{m} \left( \sum_{k \in \mathcal{Y}_{m}} \Theta_k(t_k,E_k)/D_{k}+c_m\right)^{-b_{m}}$ is convex {with respect to $\{t_{k}|\forall k\}$ and $\{E_{k}|\forall k\}$} .
Because a point-wise maximum over a group of convex functions is convex, the objective function of P1 is convex.
Adding to the fact that constraints \eqref{P_ca}, \eqref{P_cc}, and \eqref{P_ce} are all linear, the only nonconvex part in {P1} is the constraint \eqref{P_cd}.

To address the nonconvexity of \eqref{P_cd}, we leverage {difference of convex programming (DCP)} \cite{mm} to replace the left side of constraint \eqref{P_cd} {with its} first order Taylor Expansion, which is an upper bound. Specifically, for any feasible solution $\vec{\mathbf{x}}^{\diamond} = \left[t_{k}^{\diamond}, E_{k}^{\diamond} \right]^{T}$ {of P1}, we define the first order Taylor Expansion of $\Theta_{k}$ as a surrogate function $\widehat{\Theta}_{k}$:
\begin{equation}
	\label{Taylor}
\widehat{\Theta}_k(\vec{\mathbf{x}}|\vec{\mathbf{x}}^{\diamond}) = \Theta_{k}\left( \vec{\mathbf{x}}^{\diamond} \right)+ \left( \vec{\mathbf{x}} - \vec{\mathbf{x}}^{\diamond} \right)^{T} \nabla \Theta_{k} \left( \vec{\mathbf{x}}^{\diamond} \right).
\end{equation}

Since $\Theta_{k}$ is a concave function, we have $\widehat{\Theta}_{k}(\vec{\mathbf{x}}|\vec{\mathbf{x}}^{\diamond}) \geq \Theta_{k} (\vec{\mathbf{x}})$.
Therefore, we can replace $\Theta_{k}$ in \eqref{P_cd} with $\widehat{\Theta}_{k}$ at a feasible point $\vec{\mathbf{x}}^{\diamond} = \{ \vec{\mathbf{t}}^{\diamond}, \vec{\mathbf{e}}^{\diamond}\}$ and formulate a surrogate problem of {\rm{P1}}.
Since the feasible set has shrunk, the surrogate problem would generate another feasible solution $\vec{\mathbf{x}}'$, which can be used to derive a tighter surrogate problem.
This leads to an iterative refinement procedure as follows: with an initial {feasible solution} $ \vec{\mathbf{x}}^{(0)} = \left\{ \vec{\mathbf{t}}^{(0)}, \vec{\mathbf{e}}^{(0)} \right\}$, we update $\vec{\mathbf{x}}^{(n)}$ at the $n$-th iteration as
\begin{equation}
	\label{surrogate_problem}
	\begin{aligned}
		&\vec{\mathbf{x}}^{(n)} = \left\{ \vec{\mathbf{t}}^{(n)}, \vec{\mathbf{e}}^{(n)} \right\}\\
		&=\arg \min_{\vec{\mathbf{t}}, \vec{\mathbf{e}}} \max_{m}~\left\{ a_{m}  \left( \sum_{k \in \mathcal{Y}_{m}} \frac{\Theta_k(t_k,E_k)}{D_{k}} +c_{m} \right)^{-b_{m}} \right\}
	\end{aligned}
\end{equation}
under the constraints of \eqref{P_ca}, \eqref{P_cc}, \eqref{P_ce} and
\begin{equation}
	\label{surrogate_constraint}
	\widehat{\Theta}_{k}(t_{k}, E_{k}|t_{k}^{(n-1)}, E_{k}^{(n-1)})/D_k \leq |\mathcal{D}_{k}|, \quad \forall k.
\end{equation}
The above problem is convex, which can be readily solved by CVX software \cite{opt1}.
The entire iterative procedure is termed DCP-based LCWRA, and its convergence property is summarized below.
\begin{theorem}
	If $\vec{\mathbf{x}}^{\diamond(0)}$ is a feasible {solution to \rm{P1}}, any limit point $\vec{\mathbf{x}}^{*}$ of the sequence $ \left\{\vec{\mathbf{x}}^{\diamond(0)},  \vec{\mathbf{x}}^{\diamond(1)}, \cdots \right\}$ is
Karush-Kuhn-Tucher solution to problem {\rm{P1}}.
\end{theorem}
\begin{proof}
It can be shown that $\widehat{\Theta}_{k}(\vec{\mathbf{x}}^{*}|\vec{\mathbf{x}}^{*}) =\Theta_{k} (\vec{\mathbf{x}}^{*})$ and $\nabla \widehat{\Theta}_{k}(\vec{\mathbf{x}}^{*}|\vec{\mathbf{x}}^{*}) =\nabla \Theta_{k} (\vec{\mathbf{x}}^{*})$.
Combining the above results and $\widehat{\Theta}_{k}(\vec{\mathbf{x}}|\vec{\mathbf{x}}^{*}) \geq \Theta_{k} (\vec{\mathbf{x}})$, and according to \cite{mm}, the theorem is immediately proved.
\end{proof}

\section{Ranking-Based LCWRA}
In the previous section, problem {\rm{P1}} is solved numerically.
In order to {obtain insights on} how learning tasks influence the resource allocation strategy, an analytical solution is derived in this section.
In particular, {it is assumed} that $E_{\rm{max}}$ is sufficiently large {such} that $E_{\rm{max}}\geq T_{\rm{max}}P_{\rm{max}}$.
Then, all users will transmit with their maximum power, i.e., $E_{k}/t_{k}=P_{\rm{max}}$, and the constraints \eqref{P_cc} and second part of \eqref{P_ca} are removed.
The data-rate is fixed to $R_k=B\log_{2} \left( 1+ P_{\rm{max}} |h_{k}|^{2}/\sigma^{2} \right)$.
Therefore, problem {\rm{P1}} with large $E_{\rm{max}}$ is given by
\begin{subequations}
	\label{P2}
	\begin{align}
		\textrm{P}2:\min_{\vec{\mathbf{t}}, u} ~~&u, \nonumber\\
		\textrm{s.t.}\quad& u\geq a_{m}\Big( \sum_{k \in \mathcal{Y}_{m}} \frac{t_{k}R_k}{D_{k}}+c_{m} \Big)^{-b_{m}}, \quad \forall m, \label{P2_ca}\\
		& \sum_{k=1}^{K} t_{k} \leq T_{\textrm{max}}, \label{P2_cb}\\
		&\frac{t_{k}R_k}{D_{k}} \leq |\mathcal{D}_{k}|, \ t_{k} \geq 0, \quad \forall k,\label{P2_cc}
	\end{align}
\end{subequations}
where $u$ is a slack variable denoting the level of generalization error.
An efficient way to determine {the optimal $u$, denoted as} $u^*$ in $\mathrm{P}2$ is to use bisection algorithm.
Specifically, given an upper bound for the searching interval, say $u_{\mathrm{max}}$, and a lower bound $u_{\mathrm{min}}$, the trial point is set to $u^\diamond=(u_{\mathrm{max}}+u_{\mathrm{min}})/2$.
If $\mathrm{P}2$ with $u=u^\diamond$ is feasible, the upper bound is updated as $u_{\mathrm{max}}=u^\diamond$; otherwise, the lower bound is updated as $u_{\mathrm{min}}=u^\diamond$.
The process is repeated until $|u_{\mathrm{max}}-u^\diamond|<\epsilon$, where $\epsilon$ is a small positive constant to control the accuracy.
Initially, we set $u_{\mathrm{max}}=1$ and $u_{\mathrm{min}}=0$.

Now the remaining question is how to provide a feasibility check of $\mathrm{P}2$ given $u=u^\diamond$.
To this end, we first minimize the transmission time via the following problem:
\begin{subequations}
	\label{P3}
	\begin{align}
		\textrm{P}3:\min_{\vec{\mathbf{t}}} ~~&\sum_{k=1}^{K} t_{k}, \nonumber\\
		\textrm{s.t.}\quad& \sum_{k \in \mathcal{Y}_{m}} \frac{t_{k}R_k}{D_{k}}\geq \left( \frac{u^{\diamond}}{a_{m}} \right)^{-\frac{1}{b_{m}}} - c_{m}, \quad \forall m,\\
		&\frac{t_{k}R_k}{D_{k}} \leq |\mathcal{D}_{k}|, \ t_{k} \geq 0, \quad \forall k,\label{P2_cc},
	\end{align}
\end{subequations}
and then check if $\sum_{k=1}^{K} t_{k}^*\leq T_{\rm{max}}$.
If so, problem $\mathrm{P}2$ with $u=u^\diamond$ is feasible;
otherwise, the time budget cannot support learning performance $u^\diamond$ and $\mathrm{P}2$ with $u=u^\diamond$ is infeasible.

The transformation from $\mathrm{P}2$ to $\mathrm{P}3$ via bisection allows us to decouple the time variables among different groups.
In other words, $\mathrm{P}3$ can be readily decomposed into subproblems $\mathrm{P}3[1],\cdots,\mathrm{P}3[M]$ {for all the learning models respectively,} and each $\mathrm{P}3[m]$ can be solved analytically by a ranking-based procedure.
{More specifically, define $U_{i,m}=R_k/D_k$, $V_{i,m}=|\mathcal{D}_k|$, and $z_{i,m}=t_k$, where $k\in\mathcal{Y}_m$ and $R_k/D_k$ is the $i$-th largest element in the set $\{R_k/D_k, k\in\mathcal{Y}_m\}$.
Then finding the optimal $t_k$ is equivalent to finding the optimal $z_{i,m}$, which can be achieved by solving the following problem.}
\begin{subequations}
	\label{P3[m]}
	\begin{align}
	\textrm{P}3[m]: \min_{\{z_{i, m}\}}\quad&\sum_{i = 1}^{|\mathcal{Y}_m|} z_{i, m}, \nonumber\\
	\textrm{s.t.} \quad& \sum_{i = 1}^{|\mathcal{Y}_{m}|} U_{i,m}z_{i,m} \geq \left( \frac{u^{\diamond}}{a_{m}} \right)^{-\frac{1}{b_{m}}} - c_{m}, \label{P3_ca}\\
	&U_{i,m}z_{i,m} \leq V_{i,m}, \quad i =1,...,|\mathcal{Y}_{m}|, \label{P3_cb} \\
	&z_{i,m} \geq 0, \quad i=1,...,|\mathcal{Y}_{m}|. \label{P3_cc}
	\end{align}
\end{subequations}
The optimal solution to $\textrm{P}3[m]$ is summarized in the following theorem.
\begin{theorem}
The optimal solution $z^{*}_{i,m}$ to $\textrm{P}3[m]$ is $z^{*}_{i,m} (u^{\diamond}) =0$ if $\left(u^{\diamond}/a_{m}\right) ^{-\frac{1}{b_{m}}} - c_{m}\leq0$. Otherwise,
	\begin{align}
	&z^{*}_{i,m} (u^{\diamond}) =
	\nonumber\\
	&
	\frac
	{\min \left(V_{i,m}, \left( \frac{u^{\diamond}}{a_{m}} \right) ^{-\frac{1}{b_{m}}} - c_{m} - \sum_{j = 1}^{i-1} U_{j,m}z^{*}_{j,m}(u^{\diamond})\right)}
	{U_{i,m}}. 	\label{P2_solution}
	\end{align}
\end{theorem}
\begin{proof}
See \cite[Appendix A]{arxiv}.
\end{proof}

According to \textbf{Theorem 2}, if $k = \arg \max_{j \in \mathcal{Y}_{m}} R_{j}$, we can set $i=1$ in \eqref{P2_solution}, which gives
\begin{equation}
\label{P2_solution_first}
t^{*}_{k} = {\frac{D_{k}}{R_k}}\min\left\{ |\mathcal{D}_{k}|, \left( \frac{u^{\diamond}}{a_{m}} \right)^{-\frac{1}{b_{m}}} -c_{m} \right\},
\end{equation}
This proves an inverse power relationship between the transmission time and the learning error $u^\diamond$.
On the other hand, if $k\neq\arg \max_{j \in \mathcal{Y}_{m}} R_{j}$, we need to compute the required number of samples $\left(u^{\diamond}/a_{m}\right)^{-\frac{1}{b_{m}}}$ for task $m$, and then subtract the number of samples transmitted by other users in group $\mathcal{Y}_m$ with higher data-rates than $R_{k}$.
In other words, the optimal time of a particular user $k$ would not be affected by the users in other groups or those in the same group but with lower data-rates than $R_{k}$.
Therefore, the entire algorithm for solving $\rm{P}2$ based on \textbf{Theorem 2} is termed ranking-based LCWRA.

\section{Simulation Results}

In this section, simulation results are provided to verify the performance of the proposed algorithms.
We consider the case of $K=4$ and $M=2$ with $\mathcal{Y}_1=\{1\}$ and $\mathcal{Y}_2=\{2,3,4\}$, where the first task is to classify MNIST dataset \cite{MNIST} via convolutional neural network (CNN), and the second task is to classify digits dataset in Scikit-learn \cite{sklearn} via support vector machine (SVM).
For CNN, it has $6$ layers (as shown in Fig.~1 in \cite{lcpa}), and the training procedure is implemented via Adam optimizer with a learning rate of $10^{-4}$ and a mini-batch size of $100$.
For SVM, it uses penalty coefficient $C=1$ and Gaussian kernel function $K(\mathbf{x}_i,\mathbf{x}_j)=\mathrm{exp}\left(-\widetilde{\gamma}\, \|\mathbf{x}_i-\mathbf{x}_j\|_2^2\right)$ with $\widetilde{\gamma}=0.001$.
The data size of each sample is $D_1=6276~\mathrm{bits}$ for the MNIST dataset and $D_2=324~\mathrm{bits}$ for the digits dataset in Scikit-learn.
It is assumed that there are $A_1=300$ CNN samples and $A_2=200$ SVM samples in the historical dataset.
The bandwidth $B=180~\mathrm{kHz}$ \cite{iot3}.
The channel of user $k$ is generated according to $\mathcal{CN}(0, \varrho)$ where $\varrho= -90~\rm{dB}$ is the path loss \cite{massive2}, and the noise power is set as $-130~\rm{dBm/Hz}$.
The simulation is repeated with $10$ independent channels, and each point in the figure is generated by averaging the $10$ monte-carlo runs.

The tuning parameters $\{a_{m}\}$ and $\{b_{m}\}$ are fitted using MNIST dataset and the digit dataset in Scikit-learn toolbox, respectively.
For CNN, we vary the sample size $v_1$ as $(v^{(1)}_{1}, v^{(2)}_{1}, \cdots) = (100, 150, 200, 300)$ and train the CNN for $5000$ iterations.
The classification errors on a test dataset with another $1000$ MNIST samples are given by $(0.2970, 0.2330, 0.2150, 0.1180)$.
For SVM, we vary the sample size $v_1$ as $(v^{(1)}_{1}, v^{(2)}_{1}, \cdots) =(30, 50, 100, 200)$.
The classification errors on a test dataset with another $797$ samples are given by $(0.4774, 0.2513, 0.2010, 0.1445)$.
Via the nonlinear least squares fitting, we have $(a_{1}, b_{1}) = (7.3, 0.69)$ and $(a_{2}, b_{2}) = (6.24, 0.72)$.
In practice, $\{a_m, b_m\}$ are estimated via extrapolation of learning curves \cite{curve1,curve2} and exact $\left\{a_m, b_m\right\}$ is not available.
To this end, we also consider the practical LCWRA scheme with estimation errors of $(a_m,b_m)$ being $10\%$, i.e., $a_{1}\sim \mathcal{U}(7.3-0.73, 7.3+0.73), b_{1}\sim \mathcal{U}(0.69-0.069, 0.69+0.069)$, $a_{2}\sim \mathcal{U}(6.24-0.624, 6.24+0.624), b_{2}\sim \mathcal{U}(0.72-0.072, 0.72+0.072)$, where $\mathcal{U}(x,y)$ represents uniform distribution within $[x,y]$.

Based on the above settings, the classification error (i.e., the larger classification error between CNN and SVM) versus the total energy budget $E_{\rm{max}} = (0.5, 1.0, 1.5, 2.0)$ in Joule at $T_{\rm{max}} = 50~\rm{s}$ is shown in Fig.~2a.
The maximum transmit power is set to $P_{\rm{max}} = (0.03, 0.06, 0.09, 0.12)$ in Watt.
Both perfect (i.e., $\{a_m,b_m\}$ involve no estimation error) and imperfect (i.e., $\{a_m,b_m\}$ involve $10\%$ estimation error) DCP-based LCWRA schemes are simulated.
Besides the proposed DCP-based LCWRA algorithm, two benchmark schemes are also simulated: 1) the throughput fairness scheme \cite{fairness}; 2) the time fairness scheme.
It can be seen that no matter the parameters $\{a_m,b_m\}$ are perfectly estimated or not, the proposed DCP-based LCWRA algorithm always achieves the minimum classification error and the error reduction could be up to $50\%$.
This is because the proposed LCWRA finds that CNN is more difficult to train than SVM, and hence allocates more resources to user $1$.
The proposed LCWRA collects $1253$ MNIST image samples, while the throughput fairness and time fairness schemes only collect $399$ and $553$ samples, respectively (when $(T_{\rm{max}},E_{\rm{max}},P_{\rm{max}}) = (50, 1, 0.06)$).
To verify the performance of the proposed ranking-based LCWRA, we consider the case of $E_{\rm{max}} =10$ and $P_{\rm{max}} = 0.03$, and the classification error versus the total time budget $T_{\rm{max}} = (25, 50, 75, 100)$ in second is shown in Fig.~2b.
It can be seen from Fig.~2b that the proposed ranking-based LCWRA reduces the classification error by at least $20\%$.

\begin{figure*}
	\centering
	\subfigure[]{
		\label{fig:subfig:b} 
		\includegraphics[width=3.3in]{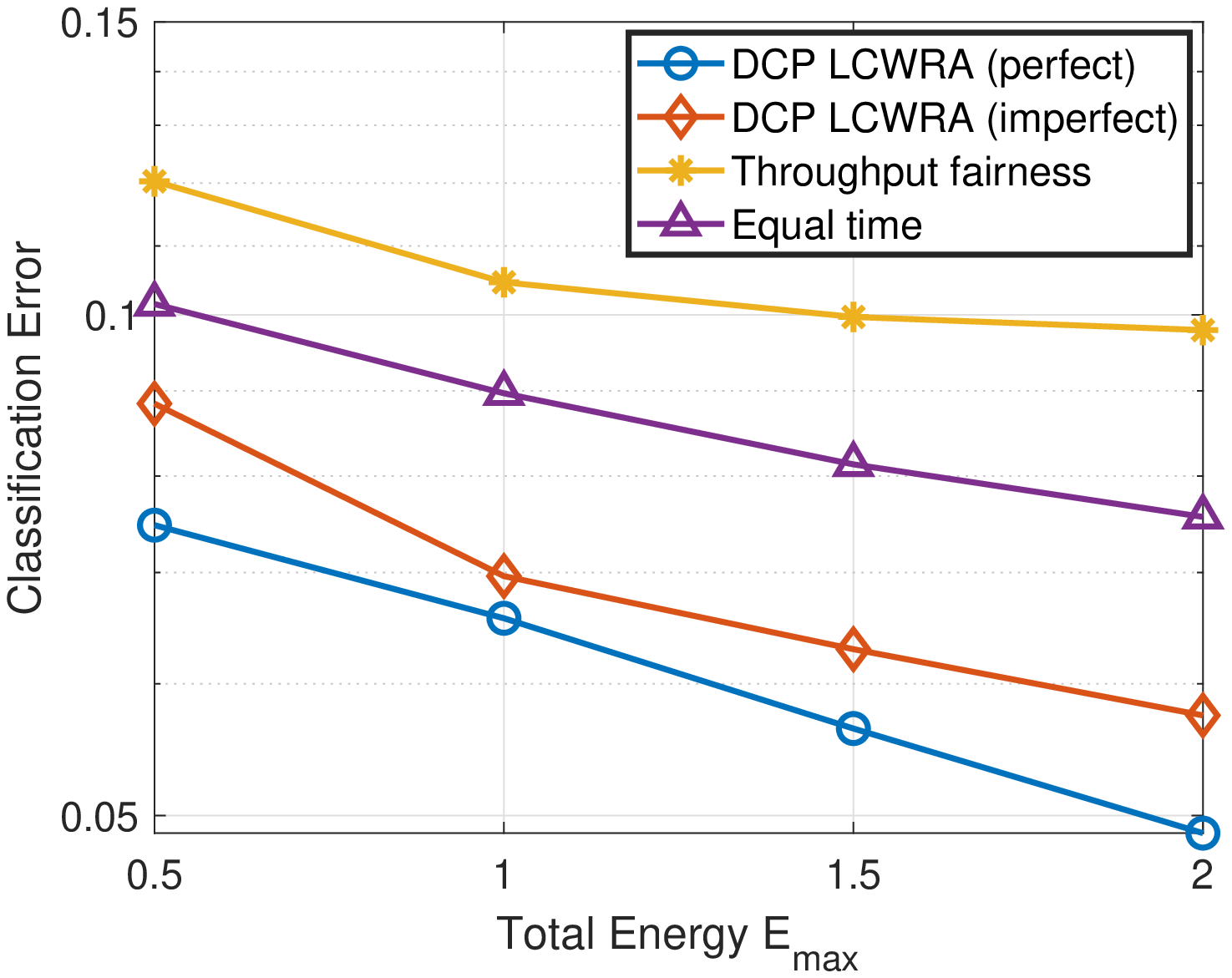}}
	\subfigure[]{
		\label{fig:subfig:b} 
		\includegraphics[width=3.3in]{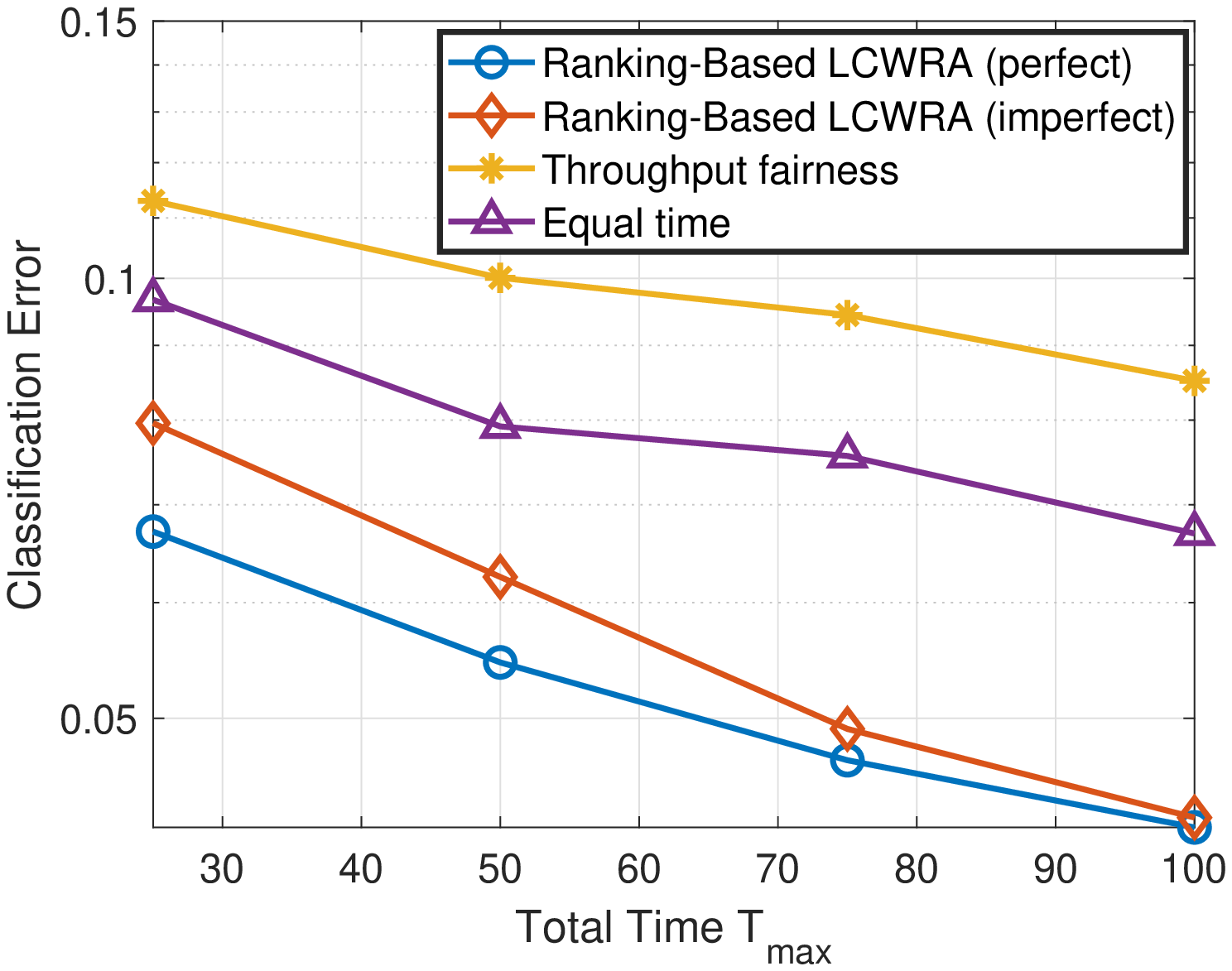}}
	\caption{Classification versus a) total energy at $T_{\rm{max}} = 50~\rm{s}$ and $P_{\rm{max}} = (0.03, 0.06, 0.09, 0.12)~\rm{W}$; b) total time at $E_{\rm{max}} = 100~\rm{J}$ and $P_{\rm{max}} = 0.03~\rm{W}$; }
	\label{simulation} 
\end{figure*}

Finally, to demonstrate the effect of the number of nodes in the IoT application, we simulate 1) the case of $K=4$ with $\mathcal{Y}_1=\{1\}$ and $\mathcal{Y}_2=\{2,3,4\}$, and 2) the case of $K=6$ with $\mathcal{Y}_1=\{1,5,6\}$ and $\mathcal{Y}_2=\{2,3,4\}$.
The worse error rate of the two tasks is 4.62\% for the case of $K=6$ and 5.23\% for the case of $K=4$.
This indicates that more IoT nodes will improve the performance of the system.
This is because more IoT nodes can generate more sensing data of various environment, and the edge could obtain a better training dataset by allocating resources to the new users $\{5,6\}$.

\section{Experimental Results}

\begin{table*}[!t]
	\caption{Experimental Results}
	\begin{center}
		\begin{tabular}{|c|c|c|c|c|c|c|c|c|c|}
			\hline
			\textbf{Scheme} & \textbf{Task}  &\multicolumn{4}{c|}{\textbf{Communication}}&\multicolumn{4}{c|}{\textbf{Learning}} \\
			\hline
			 Name & Name &\tabincell{c}{Time \\ (s)} &\tabincell{c}{Rate \\ (sample/s)} &\tabincell{c}{Total Numbr\\of Samples} &\tabincell{c}{Number\\of Samples} &\tabincell{c}{Online \\  Accuracy} &\tabincell{c}{Minimum \\ Online \\ Accuracy} &\tabincell{c}{Offline \\ Accuracy} &\tabincell{c}{Minimum \\ Offline \\ Accuracy} \\
			\hline
			\multirow{2}*{Time Fairness}&
			CNN& $30$& $5$
			&\multirow{2}*{450}
			&140& 65.3
			&\multirow{2}*{65.3}& 75.5 &\multirow{2}*{75.5}\\
			\cline{2-4} \cline{6-7} \cline{9-9}
			&SVM& $30$& $10$
			& &410 &94.97 & &94.97 & \\
			\hline
			\multirow{2}*{LCWRA}&
			CNN& $48$& $5$
			&\multirow{2}*{\tabincell{c}{360 \\$(-20\%)$}}
			&210& 74.5
			&\multirow{2}*{\tabincell{c}{74.5 \\$(+9.2\%)$}}& 81.1 &\multirow{2}*{\tabincell{c}{81.1 \\$(+5.6\%)$}}\\
			\cline{2-4} \cline{6-7} \cline{9-9}
			&SVM& $12$& $10$
			& &230 &85.5 & &85.5 & \\
			\hline
	\end{tabular}
		\label{tab1}
	\end{center}
\end{table*}

Due to the transmission overhead of control signals (which lead to reduced number of useful data) and random nature of wireless channels (e.g., disconnection/reconnection procedures), there is a mismatch between the simulated environment and the real-word environment.
Therefore, in this section, two experiments to analyze the robustness of our scheme against practical uncertainties and complex environment are elaborated.

\subsection{Robotic Image Dataset Collection}

As shown in Fig.~3a, our first experiment consists of an {mobile} edge server and two IoT users.
The edge server consists of a mobile robot (Turtlebot), a robot motion controller (Intel NUC), a computer (ThinkPad laptop), and a WiFi access point (Huawei 5G CPE Pro).
The IoT user consists of a data storage module (Raspberry Pi 3 Model B) and a WiFi transmitter (with a $60~\rm{dB}$ attenuator to simulate the low transmit power of IoT devices).
There are $1000$ MNIST data samples for training CNN and $1000$ digit samples for training SVM at user 1 and user 2 respectively.

All the functions are implemented using Python on a Ubuntu operating system.
The communication between the edge and the IoT users is established using the IP address and the socket framework.
A model training function (implemented based on Tensorflow), a data collection function (which tries to receive the wireless packets) and a buffer control function (which stores the collected data in a queue and updates the data batch before each training iteration) are developed at the edge server.
A data transmission function (which tries to send the wireless packets) and a connection/reconnection function (which periodically searches for the target edge) are developed at each user.

In our experiment, the ranking-based scheme is implemented, which allocates the communication time according to \textbf{Theorem 2}.
The total communication time is set to $1$ minute.
The data-rates at users are fixed to $R_{1} = 5D_{1}~\rm{bps}$ and $R_{2} = 10D_{2}~\rm{bps}$ (i.e., user $1$ transmits $5$ samples per second and user $2$ transmits $10$ samples per second).
The experiment is conducted in an office building shown in Fig.~3b, where user 1 is placed in the conference room and user 2 is placed in the corridor.
The edge server has no historical data (i.e., $c_{1} = c_{2} = 0$) and it cannot connect to both users at the initial location.
As a result, the robot needs to move according to the trajectory defined in the ROS of Intel NUC, and communicate with different users at different locations.
In our experiment, the server first enters the conference room to collect MNIST samples from user $1$, and then enters the corridor to collect digit samples from user $2$.
Based on the above setup, the experimental results are shown in Table~I.
Compared with the time fairness scheme, the proposed LCWRA scheme improves the online accuracy (i.e., we test the trained models immediately after data collection) and the offline accuracy (we train the learning models until convergence and then test the models) by $9.2\%$ and $5.6\%$, respectively.
The experiment demonstrates that the proposed LCWRA is robust against the uncertainty factors in real environment.

\subsection{Vehicular Point-Cloud Dataset Collection}

\begin{figure*}
	\centering
	\subfigure[]{
		\label{fig:subfig:b} 
		\includegraphics[height=2.5in]{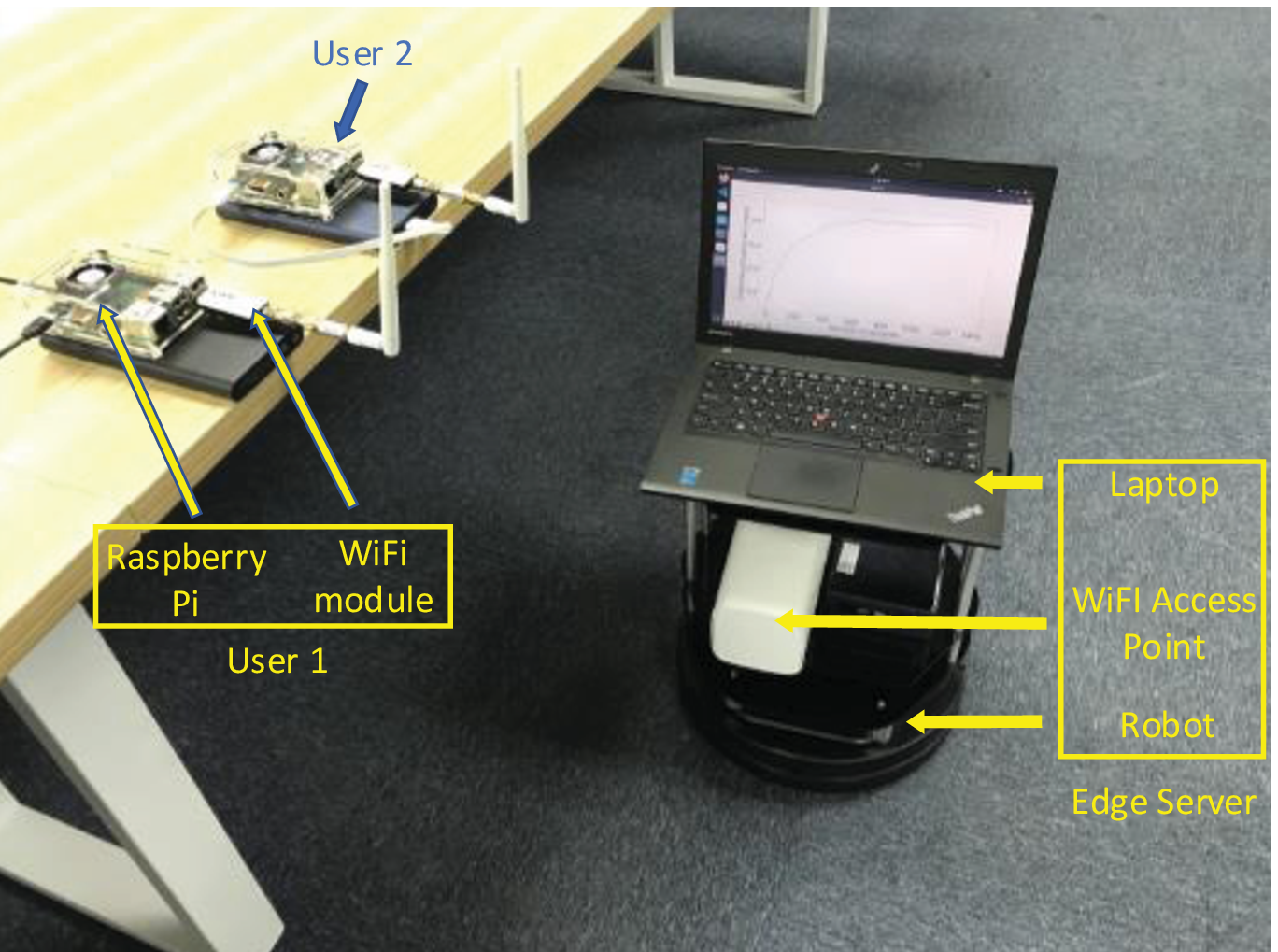}}
	\subfigure[]{
		\label{fig:subfig:b} 
		\includegraphics[height=2.5in]{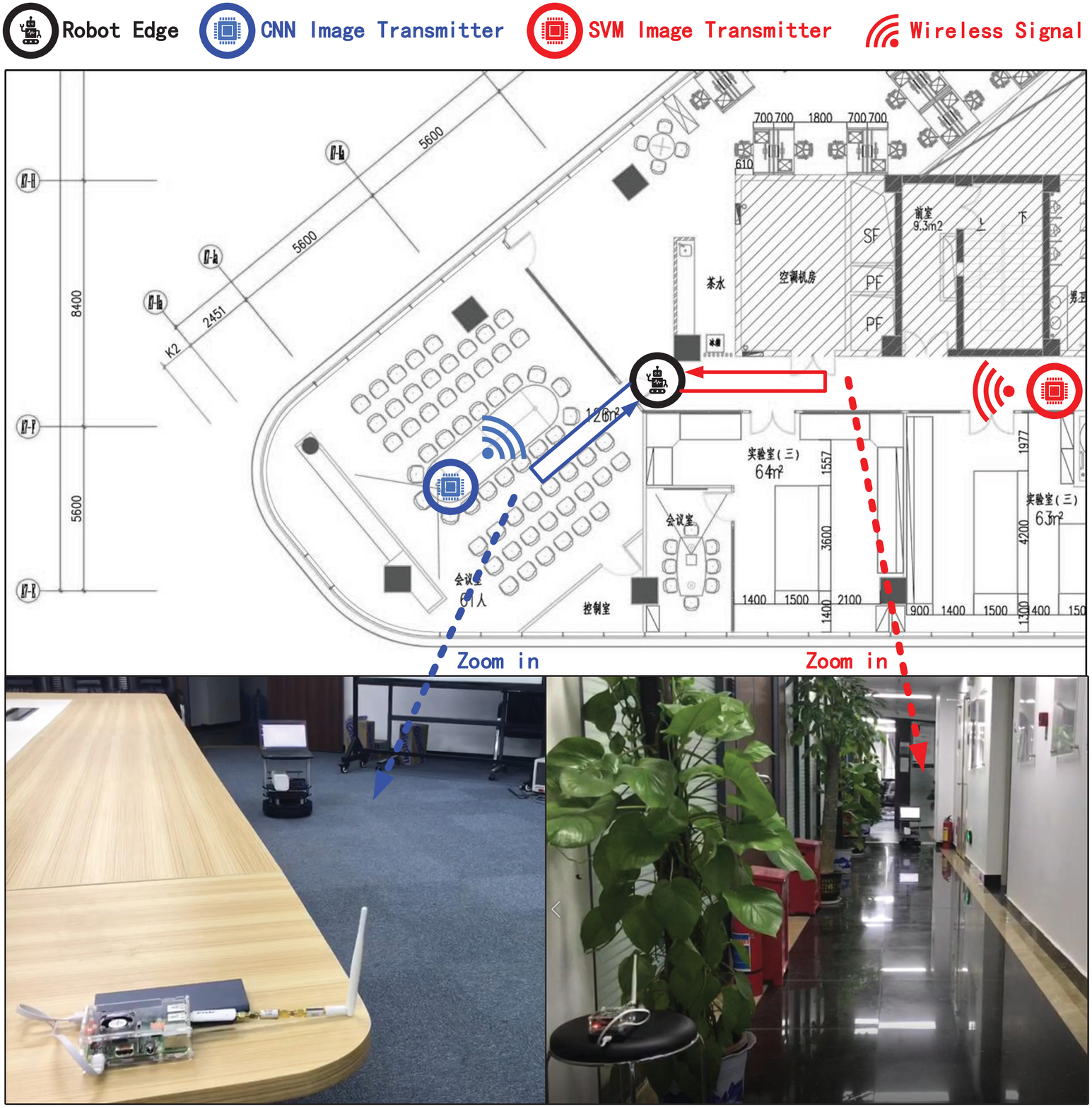}}
	\caption{a) Experimental setup of the edge computing system; b) the edge data collection and training procedure.}
	\label{simulation} 
\end{figure*}

\begin{figure*}
	\centering
	\subfigure[]{
		\label{fig:subfig:task1_lcta} 
		\includegraphics[width=1.7in]{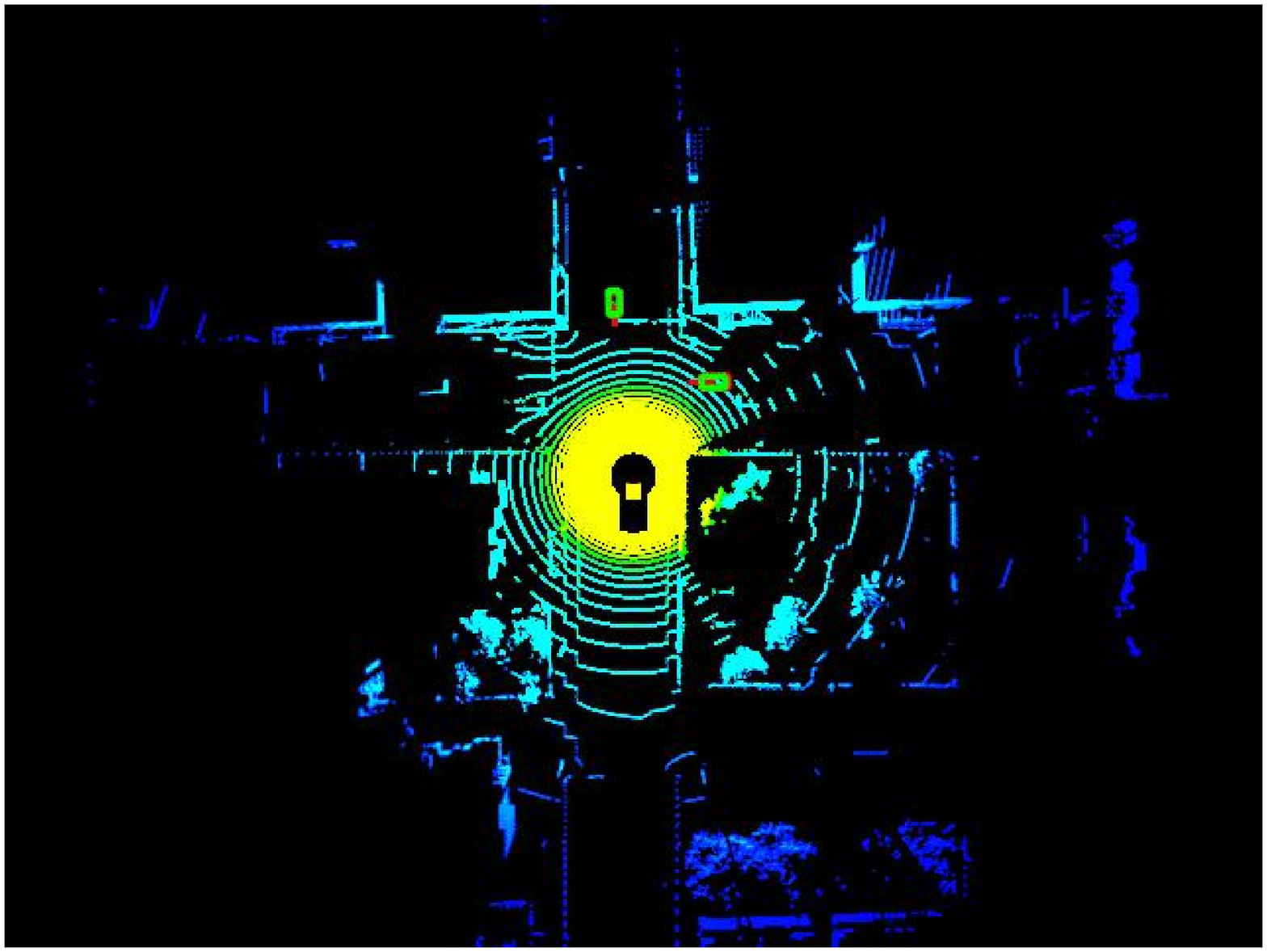}}
	\subfigure[]{
		\label{fig:subfig:task1_equal} 
		\includegraphics[width=1.7in]{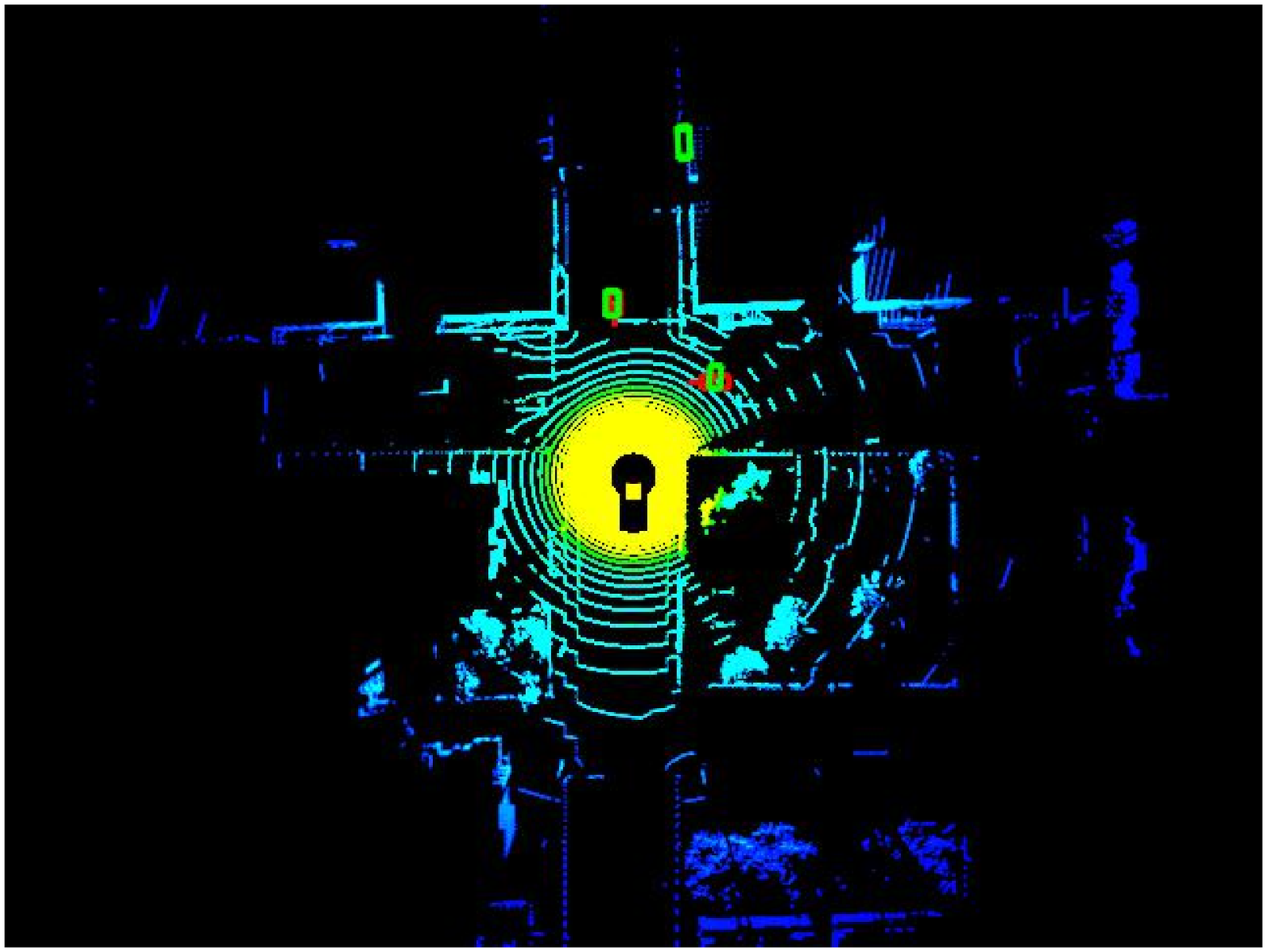}}
	\subfigure[]{
		\label{fig:subfig:task2_lcta} 
		\includegraphics[width=1.7in]{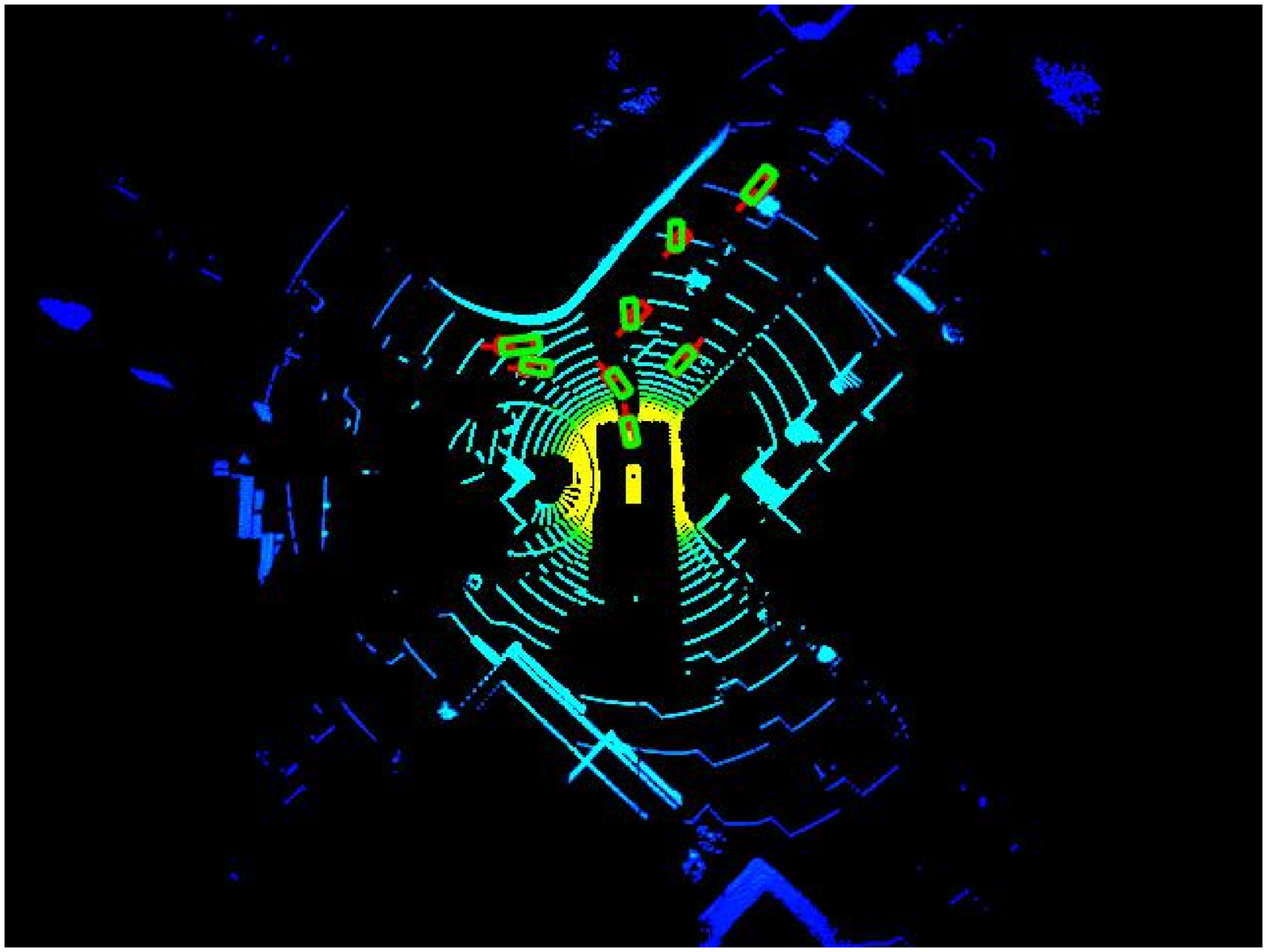}}
	\subfigure[]{
		\label{fig:subfig:task2_equal} 
		\includegraphics[width=1.7in]{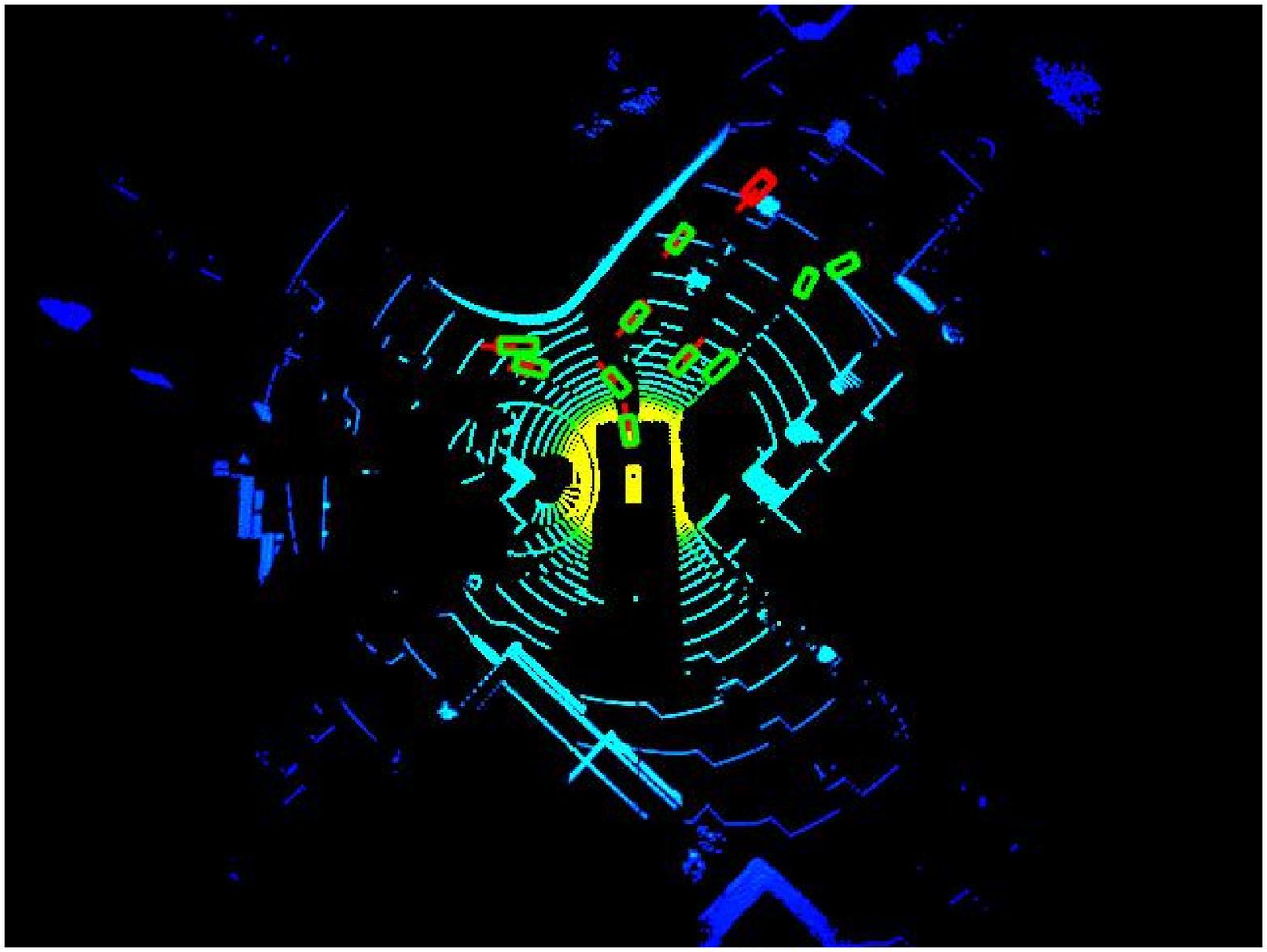}}
	\caption{Object detection results: a) task 1 with the LCWRA scheme; b) task 1 with the equal time allocation scheme; c) task 2 with the LCWRA scheme; d) task 2 with the equal time allocation scheme.}
	\label{carla} 
\end{figure*}

To verify the proposed LCWRA scheme in a more complex environment, we consider the use case of vehicular point-cloud dataset collection in Section II.
Specifically, the case of $K=2$ with two different traffic scenarios is simulated: 1) sparse traffic scenario, and 2) dense traffic scenario.
In each scenario, an autonomous driving car senses the environment, and generates on-board LiDAR point clouds.
The sensing data is then uploaded to a network edge for model training.
It is assumed that the transmission rates (controlled by the edge) at vehicles are fixed to $R_1=R_2=10~$samples/s.
The total transmission time is $T_{\mathrm{max}}=16~$s.
Notice that due to different traffics and environments, the learning error models for the two scenarios are different.

To simulate the two scenarios mentioned above, we employ the CARLA platform \cite{C5} for dataset generation and the SECOND neural network \cite{A5,shaoshuai} for object detection.
Specifically, CARLA is an open-source simulator that supports development, training, and validation of autonomous driving systems.
On the other hand, SECOND net is a voxel-based neural network that converts a point cloud to voxel features, and feeds the voxels into feature extraction and region proposal networks.
As CARLA and SECOND are not compatible, we develop a data transformation module, such that the transformed dataset satisfies the KITTI standard \cite{A6}.

Based on the above platform, we train the SECOND network for 50 epochs with a learning rate between $10^{-4}$ and $10^{-6}$ under different number of samples.
Following similar procedures in Section V, the parameters in the error models for task 1 and task 2 are given by (3.95, 0.5) and (3.11, 0.71), respectively.
Then, by executing our proposed algorithms, the numbers of collected data samples for the two tasks are obtained as 137 and 22.
In contrast, with the equal-time allocation scheme, $80$ samples are collected for both task 1 and task 2.
The trained models are tested on two validation datasets, each with 193 unseen samples ($>1000$ objects).
The worse learning accuracies of the two tasks are $67.3\%$ for LCWRA algorithm and $46.84\%$ for equal time scheme, respectively.
It can be seen from the Fig.~4 that by using the LCWRA algorithm, the performance of task 1 is significantly improved while the performance of task 2 is still acceptable.

\section{Conclusions}

This paper introduced the LCWRA scheme for edge intelligence systems.
The joint energy and time allocation scheme was derived based on DCP and the asymptotic solution proved the inverse power relationship between the generalization error and the transmission time.
Simulation results showed that the proposed LCWRA scheme significantly outperforms existing resource allocation schemes.
Experiments demonstrated the robustness of the proposed algorithm in practical and complex environment.

\appendices
\section{Proof of Theorem 2}

First, if $\left(u^{\diamond}/a_{m} \right)^{-\frac{1}{b_{m}}} - c_{m}\leq 0$, the first constraint of $\textrm{P3[m]}$ is always satisfied and $\textrm{P3[m]}$ is equivalently written as
\begin{subequations}
	\begin{align}
	\min_{\{z_{i, m}\}}\quad&\sum_{i = 1}^{|\mathcal{Y}_m|} z_{i, m}, \nonumber\\
	\textrm{s.t.} \quad&U_{i,m}z_{i,m} \leq V_{i,m}, \quad i =1,...,|\mathcal{Y}_{m}|,  \\
	&z_{i,m} \geq 0, \quad i=1,...,|\mathcal{Y}_{m}|.
	\end{align}
\end{subequations}
It is clear that the optimal solution to the above problem is $z_{i, m}^*=0$.

On the other hand, if $\left( u^{\diamond}/a_{m} \right)^{-\frac{1}{b_{m}}} - c_{m}>0$, the Lagrangian of $\textrm{P3[m]}$ is given by
\begin{align}
\label{lag}
L&= \sum_{i = 1}^{|\mathcal{Y}_m|} z_{i, m}
+ \lambda_{m} \left( \left( \frac{u^{\diamond}}{a_{m}} \right)^{-\frac{1}{b_{m}}} - c_{m} - \sum_{i = 1}^{|\mathcal{Y}_m|} U_{i,m} \, z_{i,m} \right) \nonumber\\
&\quad{}{}
+ \sum_{i = 1}^{|\mathcal{Y}_m|} \beta_{i,m} \left( U_{i,m} \, z_{i,m} - V_{i,m} \right)
- \sum_{i=1}^{|\mathcal{Y}_m|} \gamma_{i,m} z_{i,m},	
\end{align}
where $\left\{ \lambda_{m}, \beta_{i,m} , \gamma_{i,m} \right \}$ are non-negative Lagrangian multipliers. According to KKT conditions, the following equations hold:
\begin{subequations}
\begin{align}
\label{dev_lag}
\frac{\partial L}{\partial z^{*}_{i,m}} = 1 - U_{i,m}\lambda_{m} + U_{i,m}\beta_{i,m} - \gamma_{i,m }  &= 0,\quad \forall i,
\\
\lambda_{m} \left( \left( \frac{u^{\diamond}}{a_{m}} \right)^{-\frac{1}{b_{m}}} - c_{m} - \sum_{i = 1}^{|\mathcal{Y}_m|} U_{i,m} \, z_{i,m} \right)&=0,
\label{dev_lag2}
\\
\beta_{i,m} \left( U_{i,m} \, z_{i,m} - V_{i,m} \right)&=0,\quad \forall i,
\label{dev_lag3}
\\
\gamma_{i,m} z_{i,m}&=0,\quad \forall i.
\label{dev_lag4}
\\
\left( \frac{u^{\diamond}}{a_{m}} \right)^{-\frac{1}{b_{m}}} - c_{m} - \sum_{i = 1}^{|\mathcal{Y}_m|} U_{i,m} \, z_{i,m} & \leq 0,\quad \forall i
\label{dev_lag5}
\\
U_{i,m} \, z_{i,m} - V_{i,m} & \leq 0,\quad \forall i,
\label{dev_lag6}
\\
z_{i,m} &\geq 0, \quad \forall i.
\label{dev_lag7}
\end{align}
\end{subequations}

Based on \eqref{dev_lag}--\eqref{dev_lag7}, we will first prove $\gamma_{1,m} = 0$ by contradiction.
Suppose that $ \gamma_{1,m} \neq 0$.
From \eqref{dev_lag4}, we obtain $z^{*}_{1,m} = 0$.
As $\left( u^{\diamond}/a_{m} \right)^{-\frac{1}{b_{m}}} - c_{m}>0$, and according to \eqref{dev_lag5}, we must have $ \sum_{i = 1}^{|\mathcal{Y}_m|} U_{i,m} \, z_{i,m} > 0$.
This means that there exists some $j$ such that $z^{*}_{j,m} > 0$.
Now, consider another solution $\{z_{i,m}'\}$:
\begin{equation}
\label{sub_t}
z_{i,m}' =
\begin{cases}
\frac{\Delta}{U_{1,m}} & i = 1 \\
z^{*}_{j,m} - \frac{\Delta}{U_{j,m}} & i = j \\
z^{*}_{i,m} & i \notin \left\{1,j\right\}
\end{cases},
\end{equation}
where $\Delta>0$ is a small constant.
Since 1) $\sum_{i}U_{i,m}z^{*}_{i,m} = \sum_{i } U_{i,m}z_{i,m}'$; 2) $\frac{\Delta}{U_{1,m}}\leq V_{i,m}$; and 3) $z_{i,m}'\geq 0$,  $\{z_{i,m}'\}$ is a feasible solution to $\textrm{P3[m]}$.
However, the objective $\sum_iz_{i,m}'=\sum_iz_{i,m}^*+\frac{\Delta}{U_{1,m}}-\frac{\Delta}{U_{j,m}}<\sum_iz_{i,m}^*$, where the inequality is due to $U_{1,m}>U_{j,m}$.
This contradicts to $\{z^{*}_{i,m}\}$ being optimal.

Next, putting $\gamma_{1,m}=0$ into \eqref{dev_lag}, we have
\begin{align}\label{beta0}
&U_{1,m}\lambda_{m} = 1+U_{1,m}\beta_{1,m}.
\end{align}
Since $U_{1,m} \beta_{1,m}\geq 0$, we have $\lambda_{m}>0$.
According \eqref{dev_lag2}, the following equation holds:
\begin{equation}
\label{lambda}
\sum_{ i = 1}^{|\mathcal{Y}_{m}|}U_{i,m}z^{*}_{i,m} = \left( \frac{u^{\diamond}}{a_{m}} \right)^{-\frac{1}{b_{m}}}-c_{m}.
\end{equation}
To derive $z^{*}_{i,m}$, we consider two cases.
\begin{itemize}
	\item If $\beta_{1,m} \neq 0$, then $z^{*}_{1,m} = \frac{V_{1,m}}{U_{1,m}}$ according to \eqref{dev_lag3}.
	In this case, $z^{*}_{1,m}\leq \frac{1}{U_{1,m}}\left( u^{\diamond}/a_{m} \right)^{-\frac{1}{b_{m}}} -\frac{c_{m}}{U_{1,m}}$ (which is obtained by putting $z^{*}_{1,m} = \frac{V_{1,m}}{U_{1,m}}$ into \eqref{lambda}).
	
	\item If $\beta_{1,m} = 0$, then $\lambda_{m} = \frac{1}{U_{1,m}}$ according to \eqref{beta0}.
	Putting this result into \eqref{dev_lag}, we have $1-\frac{U_{i,m}}{U_{1,m}}+U_{i,m}\beta_{i,m}-\gamma_{i,m} = 0$.
	But since $U_{1,m} > U_{i,m} >0$ for any $i>1$ and $\beta_{i,m}\geq 0$, we obtain $\gamma_{i,m} >0$, meaning that $z^{*}_{i,m} = 0$ for any $i >1$.
	Further due to \eqref{lambda}, we have $z^{*}_{1,m} = \frac{1}{U_{1,m}}\left(u^{\diamond}/a_{m} \right)^{-\frac{1}{b_{m}}} -\frac{c_{m}}{U_{1,m}}$.
	In this case, $z^{*}_{1,m}\leq \frac{V_{1,m}}{U_{1,m}}$ (which is obtained from \eqref{dev_lag6}).

\end{itemize}

Based on the above discussions, $z^{*}_{1,m}$ is given by
\begin{align}
\label{opt_t1}
z^{*}_{1,m} = \min\left\{ \frac{V_{1,m}}{U_{1,m}}, \, \frac{1}{U_{1,m}}\left( \frac{u^{\diamond}}{a_{m}} \right)^{-\frac{1}{b_{m}}} -\frac{c_{m}}{U_{1,m}} \right\}.
\end{align}

With the optimal $z^{*}_{1,m}$, the optimal $z^{*}_{j,m}$ with $j>1$ can be obtained by induction.
More specifically, assuming that the optimal $z^{*}_{1,m},\cdots,z^{*}_{j-1,m}$ is known at the $j$-th iteration, and defining
$S_{j,m}=\left(u^{\diamond}/a_{m}\right)^{-\frac{1}{b_{m}}} - c_{m}-\sum_{i = 1}^{j-1} U_{i,m}z_{i,m}^*$, problem $\textrm{P3[m]}$ is equivalently written into
\begin{subequations}
	\begin{align}
	\min_{\{z_{i, m}\}}\quad&\sum_{i= j}^{|\mathcal{Y}_m|} z_{i, m}, \nonumber\\
	\textrm{s.t.} \quad& \sum_{i = j}^{|\mathcal{Y}_{m}|} U_{i,m}z_{i,m} \geq S_{j,m}, \\
	&U_{i,m}z_{i,m} \leq V_{i,m},\quad i=j,...,|\mathcal{Y}_{m}|,  \\
	&z_{i,m} \geq 0, \quad i=j,...,|\mathcal{Y}_{m}|.
	\end{align}
\end{subequations}
It can be seen that the above problem has the same form as $\mathrm{P}3[m]$.
Furthermore, $U_{j,m}$ is the largest among $\{U_{j,m},\cdots,U_{|\mathcal{Y}_{m}|,m}\}$.
Therefore, the optimal $z^{*}_{j,m}$ can be obtained following a similar derivation in \eqref{opt_t1} and the proof is completed.

\end{document}